\documentclass[notitlepage,runningheads,a4paper]{llncs}

\usepackage{epsfig}
\usepackage{amssymb,amsmath}
\usepackage{mathrsfs}
\usepackage{graphicx}
\usepackage{multirow}
\usepackage{url}
\usepackage[nodayofweek]{datetime}
\usepackage{color}
\usepackage{framed}
\usepackage{pbox}
\usepackage{newcent}
\usepackage{minitoc} 
\usepackage{float}
\usepackage{anyfontsize}
\usepackage{pbox}
\usepackage{float}

 \newcommand{\remove}[1]{}

\mathchardef\myhyphen="2D

\newcommand{\bI}{{\mathsf I}}
\newcommand{\bH}{{\mathsf H}}
\newcommand{\bPr}{\mathsf{Pr}}

\newcommand{\cR}{{\cal R}}
\newcommand{\cC}{{\cal C}}
\newcommand{\cO}{{\cal O}}
\newcommand{\cM}{{\cal M}}
\newcommand{\cW}{{\cal W}}

\newcommand{\cS}{{\cal S}}

\newcommand{\bF}{{\mathbb F}}
\newcommand{\SD}{{\bf SD}}

\newcommand{\awtppd}{$\mathsf{AWTP_{PD}}\; $}
\newcommand{\awtp}{$\mathsf{AWTP}\; $}
\newcommand{\pd}{$\mathsf{PD}\;$}
\newcommand{\ext}{\mathsf{Ext}}

\newcommand{\cA}{{\cal A}}

\newcommand{\sfPD}{\mathsf{PD}}
\newcommand{\sfSMT}{\mathsf{SMT}}
\newcommand{\rc}{$\mathsf{RC}_m$}

\newcommand{\rw}{{$(\rho_r, \rho_w)$}}
\newcommand{\ed}{{$(\epsilon, \delta)$}}

\newcommand{\rorw}{{ (\rho_r, \rho_w) }}
\newcommand{\smtone}{{$(\epsilon, \delta)$-SMT$^{[ow-s]}$-PD }}
\newcommand{\smtonefam}{{$\epsilon$-SMT$^{[ow-s]}$-PD }}
\begin{document}

\title{Adversarial Wiretap Channel with Public Discussion}

\author{Pengwei Wang and Reihaneh Safavi-Naini}

\maketitle

\begin{center}
\end{center}

\begin{abstract}
Wyner's elegant model of  wiretap channel exploits noise in the communication  channel to provide perfect secrecy against a computationally unlimited eavesdropper without requiring a shared key.
We consider an adversarial model of wiretap channel proposed in   \cite{PS13,PS14}  where the  adversary is active: it selects a fraction 
$\rho_r$  of  the transmitted codeword to eavesdrop and   a   fraction $\rho_w$ of the codeword  to corrupt 
by ``adding" adversarial error. It was shown that this model  also  captures  network 
adversaries in the setting of 1-round Secure Message Transmission \cite{DDWY93}.  It was proved that  
secure communication (1-round) is possible if and only if   $\rho_r + \rho_w <1$.
   
In this  paper  we show that by allowing communicants  to have access to a public discussion channel (authentic communication without secrecy) secure communication becomes possible even if  $\rho_r + \rho_w >1$. We  formalize the model of \awtppd protocol and for two  efficiency measures, {\em  information rate } and  {\em message round complexity}   derive   tight  bounds.  
We also construct  a rate optimal protocol family 
with minimum number of message rounds. 
We show application of  these results to Secure Message Transmission with Public Discussion (SMT-PD), and in particular show a new lower bound on transmission rate of these protocols together with a new construction of 
an optimal SMT-PD protocol.  
\end{abstract}


\section{Introduction}

 In Wyner's  \cite{W75}  model of secure communication and its generalization to broadcast scenario \cite{CK78}, Alice is connected to Bob and Eve  through two noisy channels, referred to
 as the {\em main channel} and the {\em eavesdropper channel}, respectively.  The goal is to send a message from Alice to Bob with perfect secrecy and reliability.
 Wyner's pioneering work showed that communication with (asymptotic) perfect secrecy and reliability is possible if the eavesdropper's channel is noisier than the main channel. Importantly, security  is information theoretic and  does not require a pre-shared secret key.  
 \remove{
 Wiretap model has attracted much attention in recent years\cite{PTDC11,CDS12,BTV12,LGP08,BB11,BBRS08} because it provides a natural model for an eavesdropping adversary in wireless communication setting {\em where an unauthorized receiver is located within the reception range of a transmitter.}
 Extensions to the wiretap model have also been widely studied
 \cite{M92,AC93,OW84,LH78,CDS12}.   Maurer \cite{M92} considered a scenario where Alice and Bob can use
 a {\em public discussion (PD) channel} in addition to the wiretap channel. A public discussion channel is an  authenticated communication channel that  can be used by the sender and the receiver, and is readable by everyone including the
 adversary. 
Maurer proved that  
in this model, secure communication is  possible
 even if the noise in the eavesdropper's channel is  lower than the main channel, thus showing
the power of the PD channel as a resource for communicants.
 }
 Adversarial model of wiretap channel where the adversary is active,  dates back to 
 Ozarow and Wyner \cite{OW84}. In their model  instead of the noise corrupting the  adversary's view of the transmissed codewprd, the adversary can select a fraction of the codeword that it would like to ``see".
More recently, wiretap channels where the active adevrsary also corrupts the communication 
have been considered 
 \cite{ALCP09,BS13,MBL09,PS13}.
In these models the adversary  can select  its view (also, observation or eavedropping) of the communication and 
  is also able to
{\em partially  jam}  the channel by injecting noise in the 
main channel.  
\remove{The models in \cite{BS13} and \cite{MBL09}  use \emph{arbitrarily varying  channels} to model the adversary.

In the arbitrarily  varying channel model, the error distributions of the main channel and the adversarial channel, depend on the channel states. The information that the adversary  observes, depends on the  state of eavesdropping channel, and the adversary can actively change the state of main channel to  disturb the transmission between the sender and receiver.
}
In this paper we consider a 
model of adversarial wiretap channel (\awtp channel) that is 
 proposed in  \cite{PS13,PS14}. 
 In this model,  the adversary adaptively chooses a fraction $\rho_r$  of the coordinates of the sent codeword for eavesdropping,
and a fraction $\rho_w$ of the codeword  to corrupt by adding an
 adversarial noise to the channel. 
The  adversary's eavesdropings  and corruptions are adaptive:   for each action the adversary uses all its observations
 and corruptions up to that point, to make its next choice.   
 The goal of the adversary is to break the security and/or reliability of  communication. Codes that provide security and reliability for these channels are
 called {\em \awtp-codes}. 
Interestingly  \awtp  model is closely related  to Secure Message Transmission (SMT) problem \cite{DDWY93} in networks where 
 Alice and Bob are connected by $N$ node disjoint  paths, a subset  of which is   controlled by a
 computationally unlimited  adversary  
and the goal is to provide secrecy and reliability for the communication.
The adversary in  \awtp channel is more general (powerful) than the widely studied threshold SMT adversary and can choose different subsets for eavesdropping and
corruption. 

\remove{
As shown in  \cite{PS13}, 
\awtp channel captures an adversarial communication setting that had
been considered for networks, referred to secure message transmission  \cite{DDWY93}.
}

\subsubsection*{Motivation}

It was proved  \cite{PS13} that 
 perfect secrecy 
 and reliability for  \awtp in 1-round communication is possible  if and only if,
$\rho_r+\rho_w<1$.  
We consider a scenario where in addition to the \awtp~channel, a public discussion channel denoted by   $\mathsf{PD}$, is
available to the communicants.
We call this model {\em \awtp 
with public discussion} (or \awtppd for short).
Our goal is to  see if the 
use of this extra resource  can make  secure communication possible when
 $\rho_r+\rho_w>1$ (for example $\rho_r=\rho_w = 0.9$).

Public discussion channels had  been considered in
wiretap  and SMT models, both.
In wiretap setting  
 it was shown  \cite{M92,AC93} that a public discussion channel  substantially expands
the range of scenarios in which secure communication is possible.  In particular  
secure communication becomes possible even if the eavesdroper channel is less noisy than the
main channel.
A similar result holds for SMT.
Access to a public discussion channel in SMT was considered
by Garay \emph{et.al.} \cite{GGO10} who  showed that 
 secure message  tranmission
will be possible when   $N\geq t+1$  while without  a \pd,   $N\geq 2t+1$.
\remove{
It has been shown that secure SMT protocol is possible 
if the number of corrupted lines $t$ is limited to $N\geq 2t+1$ where $N$ is the total number of paths, 
and for 1-round perfect secrecy we must have $N\geq 3t+1$.
    }

We allow communicants to interact over the \pd
but assume {\em communication over the \awtp channel is one-way} and  from Alice to Bob.
This restriction is to simplify our analysis and as we will show, will still allow us to construct protocols that are optimal.
The assumption is  also natural 
in settings where   the sender node is  more powerful such as a base station.
\remove{
and can
analysis of the protocol is primarily motivated by 
This is mainly to  We consider the one-way communication of AWTP channel  due to analyze easily and leave two-way communication over AWTP channel as future work. The one-way communication model is also applicable to practical communication, such as television and radio broadcast system.
We  assume communication over $\mathsf{PD}$ is two-way.}

{\em Our results are self-contained and 
\cite{PS13,PS14} are used motivate the study 
of the \awtp~model with  $\mathsf{PD}$. 
}

\subsection{Our work}

\subsubsection{Model and Definitions}
We define a multi-round {\em message transmission protocol} over $\mathsf{AWTP_{PD}}$.
\remove{
 and formalize  its security and reliability 
in line with the corresponding definitions of \awtp codes \cite{PS13}. 
}
The protocol may leak information to the adversray and the decoder  may output 
 an incorrect message.
We define secrecy as  the statistical distance between  the adversary's view of any two adversarially chosen messages, and reliability as
the probability that  the decoded message being different from the sent one, for any message.

An \awtppd protocol in general, has multiple {\em message rounds} where in each message round a {\em protocol message} is sent by  Alice  
over \awtp channel or the \pd channel, 
or by Bob 
over the \pd channel, each message possibly of different length.
 In each  invocation of the \awtp channel the adversary can choose a different read and write set.
An $(\epsilon, \delta)$-\awtppd protocol guarantees that the leaked information  about the message is bounded by $\epsilon$, and the probability of decoding an incorrect message is bounded by $\delta$. 
The information  {\em rate $\mathsf R$} of  a \awtppd protocol  measures transmission efficiency of
the protocol in terms of  transmission over  the \awtp channel and  is the number of message (information) bits transmitted by the protocol, divided by the total number of transmitted bits over this 
channel.    The secrecy capacity $\mathsf{C}^\epsilon$ 
of  an \awtppd channel is the maximum   information rate that can be achieved  
by a \awtppd protocol family as the total number of bits  communicated over the \awtp channel goes to infinity when the security loss is bounded by $\epsilon$.

\subsubsection{Bounds}
We derive a tight upper bound on  $\mathsf{R}$: 
we first derive a bound on $\bH(M)$,
and then use the bound to prove that the highest  secrecy rate of an \ed-\awtppd protocol
is bounded by $\mathsf{C}^{\epsilon}\leq 1-\rho+2\epsilon\cdot (1+\log_{|\Sigma|}\frac{1}{\epsilon})+2\epsilon n$, where $n$ is the total (bit) length of transmission over the \pd channel,  $\Sigma$ is the alphabet of the AWTP channel,
and $\rho = \frac{1}{N}  |S_r\cup S_w|$ is the 
fraction of components of a codeword  that  are read or written to, by the adversary. For perfect secrecy capacity we have $\mathsf{C}^{0}\leq 1-\rho$. 
When  $S_r\cap S_w\neq \emptyset$, we have $\rho< \rho_r+\rho_w$, 
 and perfectly secure communication {\em is } possible even if $\rho_r +\rho_w >1$ (e.g. $\rho_r =\rho_w =0.9$), as long as
 $\rho<1$. 

A second efficiency measure 
is the {\rm message round complexity} \rc ~of the protocol. 
We derive a  tight lower bound on  \rc  ~
for any \awtppd protocol (one-way communication over \awtp) with positive 
rate, when $\rho_r+\rho_w>1$. 
 We show that a secure  \awtppd protocol  with $\rho_r+\rho_w>1$ and $\rho<1$, cannot have two message rounds and so \rc $\geq 3$.  
 \remove{
 {\color{red} construct a protocol with this message  round complexity.}

first  showing that  
no \awtppd protocol  with $\rho_r+\rho_w>1$, and $\rho<1$,  
can have $\mathsf{PD}$ be used only by Bob, or only by Alice.
It follows that,  the \awtppd protocol must have at least three rounds:  one round of  transmission over \awtp channel, one round of transmission over \pd channel by Bob, and one round of \pd channel by Alice.
}

\subsubsection{Construction of \awtppd protocol}
We construct a family of three message round  $(0,\delta)$-\awtppd protocols
for which the rate  can be made arbitrarily close to the
upper bound. 
That is, for any small $\xi>0$, there is  $N_0$, such that  for all $N>N_0$, the rate of the \awtppd protocol  family satisfies, $\mathsf{R}\geq 1-\rho-\xi$   and so the family achieves the capacity. 
The number of message rounds of the protocol is minimal and meets the lower bound on~\rc. 
The construction   is as follows: in the first message round Alice sends to Bob over  the \awtp channel a random sequence  over $\Sigma$.
In the second message round, Bob  randomly chooses elements of 
a 
 universal hash family to   calculate
the hash values of each of the received elements,  and sends the hash values together with the
randomness used when choosing the hash function, to Alice over the \pd channel.
In the third message round,  Alice, encrypts the message using a key that is extracted from the
random values that are correctly received by Bob and sends it over the \pd channel to Bob, together with sufficient information that allows Bob to calculate the same key and recover the message.

\subsection{Relation with SMT-PD}

In \emph{secure message transmission with public discussion channel (SMT-PD)} \cite{GGO10}, in additions to wires, 
\remove{ Alice and Bob are connected by $N$ node disjoint communication paths in a network, a subset  of which can be   controlled by a
 computationally unlimited  adversary,  
 and also an authenticated  public discussion channel that can be read by everyone. The adversary chooses a subset 
 of  wires  and corrupts them arbitrarily.}
 communicants have  access to a \pd.
 Efficieny of  SMT-PD protocols is in terms of {\em transmission rate} (number transmitted bits over wires for each message bit).

 Previous works on \awtp showed  correspondence between a  1-round {\em symmetric} SMT protocol and a \awtp code.
 A symmetric SMT protocol requires the set of transcripts on all wires to be the same.  All known threshold SMT protocols are symmetric. In the rest of this paper we use the term SMT to refer to  symmetric SMT protocols.
 In Section \ref{sec_smtawtp} we  define  \smtone,  a subset of SMT-PD protocols in which only Alice can send protocol messages over the 
 wires but \pd can be used in both ways.
The bounds and the construction of \awtppd
result in a  lower bound  on the transmission rate,  a lower bound on the  message round complexity, and  a new construction for  \smtone. 
In Section \ref{sec_smtawtp} we compare these results with the known  bounds and constructions of SMT-PD.
The  message round lower bound for \smtone    also lower bounds  the message round complexity  of general SMT-PD (two-way communication over wires) and
so can be compared with the  round complexity bounds in \cite{GGO10,SJST09}. Similarly the construction of \smtone can be compared with 
those in \cite{GGO10}.
A detailed comparison of the constructions  is given in Table 1. 
Compared to other  SMT-PD protocols   that achieve  
 the upper bound on the   information
 rate of an  \smtonefam  family when the number of wires ($N$) grows while  the fraction of eavesdropped
  and corrupted wires are given by  the constants $\rho_r$ and $\rho_w$ respectively, and the leakage is  bounded by $\epsilon$, 
the unique property  of our construction is that {\em the adversary's eavesdropping and corruption sets can be different.}

\remove{t is also the only construction whose information rate achieves the upper bound on the asymptotic  information
 rate of an SMT protocol family when the number of wires ($N$) grows, while  the fraction of eavesdropped
  and corrupted wires are given by  the constants $\rho_r$ and $\rho_w$, respectively.

We consider {\em information rate} of  {\color{red} an SMT protocol  which is the inverse of the transmission rate (ratio of number of message bits
to the number of bits sent over the wires),  as a new efficiency measure for these protocols and use it to define 
 the information capacity of an SMT-PD family when connectivity ($N$)) grows while  the fraction of eavesdropped
  and corrupted wires remain  constant. 
  
We show that 
our  \smtone protocol  is the only SMT-PD protocol that achieves this capacity.} 

}

\subsection{Related Work}
Maurer's  \cite{M92} introduced  \pd channels  first in the context of 
\emph{key agreement} over wiretap channels; this was also independently considered in \cite{AC93}. 
Since  the  \pd channel is considered  free, the established key can be used to send the message securely over this
channel and so
 the communication cost of the message transmission will stay the same as 
 that of the key establishment.
 \remove{
In a key agreement protocol,  the goal of Alice and Bob is to generate a common random key.

The scheme is information-theoretic security. Key agreement protocol in \cite{M92} is over the binary symmetic channel and \cite{AC93} is over discrete memoryless channel. The adversary is passive and only evasdrop the transmission, which is same as wiretap channel with passive adversary. 

In wiretap channel with public discussion, once a shared key is established, 
secure message transmission can be achieved by encrypting the message and sending it over the \pd
channel.  
Since the \pd channel can be freely used,
 the communication cost of the message transmission will stay the same as 
 that of the key establishment. 
}
  Our construction   also has two steps: a key establishment, 
followed by encrypting the message and sending it  over the public discussion channel.
This is also the approach in \cite{GGO10} (Protocol I) and \cite{SJST09}.

The model of adversarial wiretap in \cite{SW13,SW131} extends  wiretap II to include active  (jamming) adversarial noise. 

 \remove{
A binary version of the model  considered in this paper was also considered in  \cite{ALCP09}.

Other models of adversarial wiretap, \cite{ALCP09,BS13,MBL09,BBS13}, and 
 their relationship to the model considered here, have been discussed in \cite{PS13}.
 The paper also establishes a relation between a  special subset of adversarial wiretap
codes and  1-round  {\em Secure Message Transmission (SMT)} \cite{DDWY93} protocols for networks. 
 In SMT,  Alice and Bob are connected by $N$ node disjoint communication paths in a network, a subset of which is controlled by an adversary who can see and arbitrarily change
the transmissions over the  controlled paths. SMT protocols provide security and reliability against this adversary.
}

SMT-PD was introduced in \cite{GO08} as a building block in almost-everywhere  secure multiparty computation.
 Bounds on the required number of rounds were derived in \cite{SJST09}. In \cite{GGO10} a bound on transmission rate over wires (not including communication over the~\pd) was derived.
 The paper presents two constructions: protocol I is optimal  in the sense that the  transmission rate is of {\em the order of the bound as the number of wire increases},
 and protocol II  in which the goal is to minimize communication over the \pd. This reduction is however 
 at the expense of lower rate on the wires. I
 Table \ref{state100} compares the information rate of these constructions for large $N$.

\subsection{Organization} 
In Section 2, we introduce \awtp channel and the \pd channel, and in Section 3, define \awtppd protocols. In Section 4, we derive
 the upper bound on the rate, and the minimum requirement on the message round complexity. 
 In Section 5, we give the  construction of an optimal \awtppd protocol. In Section 6, we give the relation between \awtppd protocol and SMT-PD protocol. 
 In Section 7, we discuss our results, open problems and future works.

\section{Preliminaries}

We use, calligraphic letters $\cal X$ to denote sets, $\bPr(X)$ to denote
a probability distribution on the set $\cal X$, and $X$ to denote a random
variable that takes values from $\cal X$ with probability $\bPr(X)$. The
conditional probability of $X$ given  $E$,  is $\bPr[X = x|E]$.
$\log()$ is logarithm in base two. Shannon entropy of a random
variable $X$ is, $\bH(X ) =  \sum_x \bPr(x) \log \bPr(x)$, and conditional
entropy of a variable $X$ given $Y$, is 
$\bH(X|Y ) =  \sum_{x,y} \bPr(x, y) \log \bPr(x|y)$. The min-entropy of  a variable $X$ is $\bH_{\infty}(X)=\min_{x\in {\cal X}}-\log \bPr(X=x)$. Statistical distance between two random variables $X_1, X_2$, defined over $\cal X$,
is given by $\SD(X_1,X_2) = \frac{1}{2}\sum_x |\bPr(X_1 = x) - \bPr(X_2 = x)|$. 
Mutual information between random variables $X$ and $Y$ is given by, $\bI(X,Y) = \bH(X)-\bH(X|Y)$.
 Hamming weight of a vector $e$ is denoted by be $wt(e)$.

\subsection{Channel Models}

We consider two types of channels: \awtp channel and \pd channel.
\remove{ In \awtp~channel, the adversary can read a fraction of transmitted information and add error on another fraction of transmitted information. In public discussion channel, the adversary can read all the transmitted information, but the transmission over public discussion channel is error free. 
}
A channel can be one-way or two-way. 
\begin{definition}
A 
{\em one-way channel} from Alice to Bob (Bob to Alice) is used to send messages 
 from Alice to Bob  (Bob to Alice). 
A 
{\em two-way channel}  can be used 
in both directions,
from Alice to Bob, or from Bob to Alice. 
\end{definition}

Let $[N]=\{1,\cdots, N\}$,   $S_r= \{i_1,\cdots, i_{\rho_rN}\} \subseteq [N]$ and  $S_w= \{j_1,\cdots, j_{\rho_wN}\} \subseteq [N]$.
Support of a vector $x= (x_1\cdots x_N) \in \Sigma^N$, denoted by $\mathsf{SUPP}(x)$, is the set of positions where  $x_i\neq 0$.

\begin{definition}A $(\rho_r, \rho_w)$-Adversarial Wiretap Channel ($(\rho_r, \rho_w)$-\awtp Channel)\label{def_1awtp} 
is an adversarial channel 
that it 
 is (partially) controlled by an adversary Eve, with two capabilities: Reading and Writing.
For a codeword of length $N$, Eve  selects a subset $S^r \subseteq [N]$ of size $|S^r|=\rho_rN$ to read (eavesdrop),
and selects a subset $S^w \subseteq [N]$ of size $|S^w|=\rho_wN$ to {\em write to} (corrupt). The writing is
by adding to $c$ an error vector $e$ with $\mathsf{SUPP}(e) = S^w$, resulting in $c+e$ to be received. 
The adversary is adaptive and to select a component  for reading and/or writing, 
it uses its knowledge of the codeword at the time.
The subset $S=S^r\cup S^w$ of size $|S|=\rho N$, is the set of components of the codeword that the adversary reads or writes to. 
\end{definition}
 The \awtp channel is called a \emph{restricted}-\awtp channel if $S_r=S_w=S$.

{\em We assume the adversarial wiretap channel is one-way and can only be used by Alice.}

\begin{definition}\label{def_pd}(Public Discussion Channel (\pd Channel)) is an 
authenticated channel between Alice and Bob, that can be 
read by everyone including Eve. 

\end{definition}

{\em We assume the \pd channel is two-way  
 can be used by Alice and Bob, both.} 

Hence in our {\em   \awtppd setting  Alice and Bob have access to a one-way \awtp channel and 
a two-way \pd channel.}
We consider protocols with multiple message rounds  and assume in each message round a message is
 sent on one of the channels available to the communicants. In particular, in each message round Alice can use 
either the \awtp or the \pd channel. 

\begin{definition}
The message round complexity $\mathsf{RC}_m$ of a protocol is the total  number invocations of channels (\awtp and \pd)
by the two 
the communicants.
\end{definition}


\section{$\mathsf{AWTP_{PD}}$ Protocol} \label{sec_awtppd}

Alice (sender) wants to send a message (information) $m\in \cM$, securely and reliably to Bob (receiver),
using a multi-round protocol over a \awtppd~channel, called an
{\em  \awtppd protocol.}

The protocol consists of a sequence of message rounds. 
Each  message round is in one of the following form: (i) Alice sends a message to Bob over \awtp channel,  (ii) Alice sends a message to Bob over \pd channel, and (iii)
Bob sends a message to Alice over the \pd channel.  
 
Let $\ell_c$ and $\ell_d$ denote the  total number of invocations of the \awtp channel, and the \pd channel, respectively, 
and  assume  
$\ell=\ell_c+\ell_d$. 
Let $r_A$ and $r_B$ denote the randomness used by Alice and Bob.

The {\em  protocol messages (also called codewords)} 
sent over the \awtp~channel and the
\pd channel are denoted by 
 $c_i$ and $d_i$, respectively.

We use $c^{i}=\{c_1 \cdots c_i\}$  to denote the concatenation of  protocol messages,
 transmitted over the \awtp  channel after the $i^{th}$ invocation of the  \awtp channel.
 Similarly  $d^{i}=\{d_1 \cdots d_i\}$ is  the concatenation of  protocol messages
sent over \pd,  after the $i^{th}$ invocation of this channel.

Let the protocol message alphabets of the \awtp  and \pd channels  be  $\Sigma$ and 
$\bF_2$, respectively. 
In  the $i^{th}$ invocation of the \awtp channel, Alice sends  a codeword of length $ N_i$. 
In the  $i^{th}$  invocation of the \pd channel, Alice or Bob, sends a binary message of length $n_i$.  The number of symbols sent over the  \awtp channel is $N=\sum_{i=1}^{\ell_c} N_i $, and the number of bits transmitted over the \pd, is $n=\sum_{i=1}^{\ell_d} n_i$.

Let the view of Alice and Bob when sending the $i^{th}$ codeword 
be,  $v^{i}_A$ and $v^{i}_B$, respectively.
The view of a participant consists of all the protocol messages that are received before sending the $i^{th}$ codeword.
When sending a message $m$,  in the $i^{th}$ invocation of the  \awtp channel, Alice constructs a codeword 
$c_i$ using her view, local randomness, and $m$, 
\[
c_i=\mathsf{AWTP_{PD}}(m, r_A, i, v^{i}_A, \mathsf{AWTP}).
\] 

In each invocation of the \pd channel, Alice (or Bob) generates the codeword $d_i$ using their view, local randomness and $m$, 
\[
d_i=\mathsf{AWTP_{PD}}(m, r_{X}, i, v^{i}_{X}, \mathsf{PD}),
\]
where $X\in \{A,B\}$ if the protocol message constructed by Alice (Bob).


\begin{definition}
[$(\epsilon, \delta)$-$\mathsf{AWTP_{PD}}$ protocol]\label{def_smt}
A secure $(\epsilon, \delta)$-$\mathsf{AWTP_{PD}}$  protocol 
satisfies the following two properties:
\begin{enumerate}
\item Secrecy:  For any two messages $m_1, m_2\in \cM$, 
the statistical distance between Eve's views of the protocol,
 when the same random coins $r_E$ are used by Eve, is bounded by $\epsilon$. 
\[
\begin{split}
\max_{m_0, m_1}{\bf SD}(&\mathsf{View_{E}}(\mathsf{AWTP_{PD}}(m_1), r_E),\mathsf{View_{E}}(\mathsf{AWTP_{PD}}(m_2), r_E))\leq \epsilon
\end{split}
\]
\item Reliability:  For any message $M_\cS$ chosen by Alice, the probability that Bob outputs the  message  sent by Alice, is at least $1-\delta$. That is,
\[
\bPr(M_\cR\neq M_\cS)\leq \delta.
\]
Here probability is over 
the randomness of Alice and Bob and the adversary. 

\end{enumerate}
\end{definition}
The \awtppd protocol provides {\em perfect secrecy} if $\epsilon=0$. If adversary is passive, then Bob can always output the correct message $m_{\cal S}$ and $\bPr(M_\cR=M_\cS)=1$. 
A {\em  restricted-$(\epsilon,\delta)$-\awtppd } protocol 
 is over a  restricted-\awtppd channel where 
 $N_i= N_j$, $S_i=S_j=S$ for any $1\leq i\leq j\leq \ell$. 

The efficiency measures of an $(\epsilon, \delta)$-$\mathsf{AWTP_{PD}}$  protocol $\Pi$ are, 
{\em  (i) the information rate $\mathsf{R}(\Pi)= \frac{\log|\cal M|}{N\log |\Sigma|}$ and, (ii) the 
message round complexity $\mathsf{RC}(\Pi)= (r_{\mathsf{awtp}}, r_{\mathsf{pd}})$ denoting the number of invocations of the \awtp and \pd channels, respectively .
}

\begin{definition}\label{def_smtfamily}
An $(\epsilon, \delta)$-\awtppd protocol family
 for
 a $(\rho_r, \rho_w)$-AWTP channel,  
is a  family  of protocols ${\bf \Pi}=\{\Pi^N\}_{N\in \mathbb{N}}$, where 
$\Pi^N= (\epsilon, \delta)$-\awtppd  is an \awtppd  protocol for the $(\rho_r, \rho_w)$-AWTP channel. 
A protocol family $\bf \Pi$ achieves {\em information  rate $\mathsf{R}$,} 
if for any $\xi>0$ there exist $N_0$ such that for any $N\geq N_0$, there is $\delta<\xi$ and,
\[
\frac{\log|\cal M|}{N\log |\Sigma|}\geq \mathsf{R}-\xi.
\]
\end{definition}
The {\em $\epsilon$-secrecy (perfect secrecy) capacity $\mathsf{C}^\epsilon$ ($\mathsf{C}^0$)} of
a \rw-\awtppd channel is the largest achievable rate of all $(\epsilon, \delta)$-$\mathsf{AWTP_{PD}}$ ($(0,\delta)$-\awtppd) protocol families
for the channel.

Note that we effectively assume communication over \pd is free and consider communication cost of the \awtp only.


\section{Bounds on \ed-$\mathsf{AWTP_{PD}}$ Protocols}\label{sec_upperbound}

We derive two bounds for  $(\epsilon, \delta)$-$\mathsf{AWTP_{PD}}$ protocols: 
an  upper bound on the rate, and 
 a lower bound on the 
minimum number of message rounds required for such protocols.



\subsection{Upper Bound on Rate}\label{sec_ratebound}
 

\begin{theorem}\label{the_smt_bound1}
The 
rate of an $(\epsilon, \delta)$-$\mathsf{AWTP_{PD}}$ protocol  is bounded by, 
\[
{\mathsf C}^{\epsilon}\leq 1-\rho+2\epsilon\cdot(1+\log_{|\Sigma|}\frac{1}{\epsilon})+ 2\epsilon n
\]
\end{theorem}

In the following proof we assume $\rho_r+\rho_w=1$, and $|S_i^r \cup S_i^w | = \rho N< N$ for $i=1,\cdots, \ell_c$. The proof can be extended to  $\rho_r+\rho_w> 1$  and $|S_i^r \cup S_i^w | = \rho N< N$ also. 
The proof outline is as follows. 
We define an adversary $\mathsf{Adv}_1$ and prove an upper bound on the rate of any  protocol over the \awtppd channel assuming this adversary.
This gives un upper bound on the rate of the 
$\mathsf{AWTP_{PD}}$ protocol  against any general adversary.
\remove{
 cannot  violate this upper bound, because $\mathsf{Adv}_1$
is one of the possible adversarial strategies that can be used against a protocol.
}

The proof has three steps.

First (Step1), we define a weak adversary that  before the start of the protocol chooses, (i)    the reading and writing sets of all invocations of the \awtp~channel, and (ii) the random errors of appropriate weight for  each \awtp channel invocation. 
For this adversary, we prove two lemmas (Lemmas \ref{le_smt_bound1} and \ref{le_smt_bound2}) related to the entropy of the transmitted message. 
 Second (Step 2), we use the lemmas to derive a bound on $\frac{\log |\cM|}{N\log |\Sigma|}$. 
Finally (Step 3) we prove the bound on the channel capacity.

\vspace{1mm}
{\em Notations.}  
Let the codeword length in the $i^{th}$ invocation of the AWTP channel  be $N_i$, 
and  $[N]=\bigcup_{i=1}^{\ell_c} [N_i]$. 
Let $S^{r}_i$ and $S^w_i$ denote the read and write sets of the adversary in the 
$i^{th}$ invocation of the AWTP channel with $|S_i^r|=\rho_rN_i$ and $|S_i^w|=\rho_wN_i$, and denote
$S^{i,r}=\{S_1^{r}, \cdots, S_i^{r}\}$ and $S^{i,w}=\{S_1^{w}, \cdots, S_i^{w}\}$.
 
 Let $S_i^{a}=S_i^{r}\backslash S_i^{w}$ be the set of read only, $S_i^{b}=S_i^{r}\cap S_i^{w}$  the set of read and write, $S_i^{c}=S_i^{w}\backslash S_i^{r}$   the set of write only, and $S_i^{d}=[N_i]\backslash (S_i^{r}\cup S_i^{w})$  the set of neither read nor write components, in the  $i^{th}$ invocation of the AWTP channel.
Finally,
$S^{{\ell_c}, a}=\cup_{i=1}^{\ell_c} S_i^{a}$, 
$S^{{\ell_c}, b}=\cup_{i=1}^{\ell_c} S_i^{b}$, $S^{{\ell_c}, c}=\cup_{i=1}^{\ell_c} S_i^{c}$,
and 
$S^{{\ell_c}, d}=\cup_{i=1}^{\ell_c} S_i^{d}$.

Let $c_i$ and $d_i$ be the codewords transmitted over the \awtp channel and \pd channel in the $i^{th}$ invocations of the two channels, respectively; $c_{i,j}$ and $d_{i,j}$ denote the $j^{th}$ components of codeword $c_i$ and $d_i$, respectively;  $c^i$ and $d^i$ denote 
concatenations of all codewords sent in all invocations up to, and including, the $i^{th}$ invocations of the 
\awtp and the \pd channels, respectively. We use capital letters to refer to the random variables associated with,
$c_i$,  $d_i$, $c_{i,j}$ ,  $d_{i,j}, c^i$ and $d^i$, as $C_i, D_i, C_{i, j}, D_{i, j}, C^{i}$ and $D^{i}$, respectively.
Let $C^{\ell_c, r}$ and $C^{\ell_c, w}$ be the random variables of the  protocol messages 
 on the sets $S^{\ell_c,r}$ and $S^{\ell_c,w}$, and
$C^{\ell_c, a}$, $C^{\ell_c, b}$, $C^{\ell_c, c}$, $C^{{\ell_c}, d}$ be the random variables corresponding to the sets,
 $S^{\ell_c, a}, S^{\ell_c, b}, S^{\ell_c, c}, S^{\ell_c, d}$, respectively.

\medskip
\begin{proof}
The proof 
has three steps:

\vspace{2mm}
\noindent
{\bf Step 1}.

We define an adversary
$\mathsf{Adv}_1$ that works as follows: 
\begin{enumerate}
\item 

Selects the reading and writing  sets   $S^{{\ell_c}, r}$ and $S^{{\ell_c}, w}$,  
 of  all  \awtp channel invocations, before the start of the
protocol. 
 
\item For each invocation, chooses a random error vector $e_i$ of appropriate weight; that is,  chooses  $e^w_i$, with uniform distribution from $\Sigma^{|S_i^w|}$; we have $\bPr(e^w_i)=\frac{1}{|\Sigma|^{\rho_wN_i}}$.
\item During the protocol execution, uses the error vectors  to corrupt the \awtp messages, 
reads the transmission on $S^{{\ell_c}, r}$ and over \pd channel.
\end{enumerate}

We give two lemmas that follow from  
$\epsilon$-secrecy and $\delta$-reliability of  the $(\epsilon, \delta)$-$\mathsf{AWTP_{PD}}$ protocol 
 against $\mathsf{Adv}_1$.
 Let $V_E$ denote the random variable of the adversary view at the end of the protocol.

\begin{lemma}\label{le_smt_bound3}
For an $(\epsilon,\delta)$-\awtppd protocol, the following holds:
\[
\bI(M; V_E)\leq 2\epsilon N\cdot \log(\frac{|\Sigma|}{\epsilon})+2\epsilon n
\]
\end{lemma}

Proof is in Appendix \ref{ap_smt_bound3}.

Since $\mathsf{Adv}_1$ selects the reading sets $S^{{\ell_c}, r}$ before the start of the 
 protocol, we have, 
$V_E=\{C^{{\ell_c}, r},D^{{\ell_d}}\}$, and so, 
we have 
\begin{eqnarray}\label{le_smt_bound1}
\bI(M;C^{{\ell_c},r}D^{{\ell_d}})\leq 2\epsilon N\cdot\log(\frac{|\Sigma|}{\epsilon})+2\epsilon n
\end{eqnarray}

\begin{lemma}\label{le_smt_bound2}
For an $(\epsilon, \delta)$-$\mathsf{AWTP_{PD}}$ protocol, the following holds assuming
  $\mathsf{Adv}_1$ adversary,
\[
\bH(M|C^{{\ell_c},a}C^{{\ell_c},d}D^{{\ell_d}})\leq \bH(\delta)+\delta\log |\cM|
\]
\end{lemma}

Proof is in Appendix \ref{ap_smt_bound2}.

Lemma \ref{le_smt_bound1} and Lemma \ref{le_smt_bound2} are used to prove an
upper bound on the rate of an $(\epsilon, \delta)$-$\mathsf{AWTP_{PD}}$ protocol, assuming  adversary $\mathsf{Adv}_1$.

\vspace{2mm}
\noindent
{\bf Step 2.}
We prove the  upper bound,
\[
\frac{\log |\cM|}{N\log |\Sigma|}\leq 1-\rho+2\epsilon\cdot(1+\log_{|\Sigma|}\frac{1}{\epsilon})+ 2\epsilon n +  2\bH(\delta) +\delta n
\]
Here, $N$ is the total number of symbols sent  over \awtp channel, and $n$ is the number of bits sent over the \pd channel. 
Let $\cC^{{\ell_c}}$ and ${\cal D}^{{\ell_d}}$ denote the set of possible protocol messages 
over the \awtp channel and the \pd channel,
respectively. 
%
We have,
\begin{equation}\label{eq_bound13}
\bH(M)=\bI(M; C^{{\ell_c}, r} D^{{\ell_d}})+\bH(M|C^{{\ell_c},r} D^{{\ell_d}})
\end{equation}
From Lemma \ref{le_smt_bound1},  the first term can be upper bound as, 
\begin{equation}\label{eq_bound12}
\bI(M; C^{{\ell_c},r}D^{{\ell_d}})\leq 2\epsilon\cdot N\log(\frac{|\Sigma|}{\epsilon})+2\epsilon n
\end{equation}
The upper bound on the second item $\bH(M|C^{{\ell_c},r}D^{{\ell_d}})$ is, 
\begin{equation}\label{eq_bound15}
\begin{split}
&\bH(M|C^{{\ell_c},r}D^{{\ell_d}})\\
&=\bH(M|C^{{\ell_c},a}C^{{\ell_c},b}D^{{\ell_d}})\\
&=\bH(MC^{{\ell_c},b}|C^{{\ell_c},a}D^{{\ell_d}})-\bH(C^{{\ell_c},b}|C^{{\ell_c},a}D^{{\ell_d}})\\
&=\bH(M|C^{{\ell_c},a}D^{\ell_d})+\bH(C^{\ell_c,b}|MC^{{\ell_c},a}D^{{\ell_d}})-\bH(C^{{\ell_c},b}|C^{{\ell_c},a}D^{\ell_d})\\
&=\bH(M C^{{\ell_c},d}|C^{{\ell_c},a} D^{\ell_d})-\bH(C^{{\ell_c},d}|M C^{{\ell_c},a} D^{{\ell_d}})+\bH(C^{{\ell_c},b}|M C^{{\ell_c},a} D^{\ell_d})-\bH(C^{{\ell_c},b}|C^{{\ell_c},a} D^{{\ell_d}})\\
&=\bH(M|C^{\ell_c,a} C^{{\ell_c},d} D^{\ell_d})+\bH(C^{\ell_c,d}|C^{{\ell_c},a} D^{{\ell_d}})-\bH(C^{{\ell_c},d}|M C^{{\ell_c}, a} D^{{\ell_d}})+\bH(C^{\ell_c,b}|M C^{{\ell_c},a} D^{{\ell_d}})\\
&\;\;\;\;-\bH(C^{{\ell_c},b}|C^{{\ell_c},a} D^{{\ell_d}})\\
&\overset{(1)}{\leq} \bH(M|C^{{\ell_c},a} C^{{\ell_c},d} D^{{\ell_d}})+\bH(C^{{\ell_c},d}|C^{{\ell_c},a} D^{{\ell_d}})-\bH(C^{\ell_c,d}|M C^{{\ell_c},a} D^{{\ell_d}})\\
&\overset{(2)}{\leq} \bH(M|C^{{\ell_c},a} C^{{\ell_c},d} D^{{\ell_d}})+\bH(C^{{\ell_c},d})
\end{split}
\end{equation}
Inequality (1) is from,\\ $\bH(C^{{\ell_c},b}|M C^{{\ell_c},a} D^{{\ell_d}})\leq \bH(C^{{\ell_c},b}|C^{{\ell_c},a} D^{{\ell_d}})$.
Inequality (2) follows from,
 $\bH(C^{{\ell_c},d}|C^{{\ell_c},a} D^{{\ell_d}})\leq \bH(C^{{\ell_c}, d})$ and $\bH(C^{{\ell_c},d}|M C^{{\ell_c},a} D^{{\ell_d}})\geq 0$.

From $\bH(C^{{\ell_c},d}) \leq \log |\cC^{{\ell_c},d}| \leq N(1-\rho)\log |\Sigma|$, we have,
\begin{equation}\label{eq_bound18}
\bH(C^{{\ell_c},d})\leq N(1-\rho)\log |\Sigma|
\end{equation}
Using Lemma \ref{le_smt_bound2}, we have, 
\begin{equation}\label{eq_bound17}
\bH(M|C^{{\ell_c},a} C^{{\ell_c},d} D^{{\ell_d}})\leq \delta\log |\cM|+\bH(\delta)
\end{equation} 
From (\ref{eq_bound15}), (\ref{eq_bound18}), (\ref{eq_bound17}), we have,
\begin{equation}\label{eq_bound11}
\bH(M|C^{{\ell_c},r} D^{{\ell_d}})\leq N(1-\rho)\log |\Sigma| + \delta\log |\cM| + \bH(\delta)
\end{equation} 
We also have,  
\begin{equation}\label{eq_bound14}
\log |{\cal M}| \stackrel{(1)}\leq 
 \log | \cC^{{\ell_c}} {\cal D}^{{\ell_d} }| \stackrel{(2)}\leq N\log |\Sigma|+n
\end{equation}
where $\cC^{{\ell_c}}{\cal D}^{\ell_d} $  are possible (error free) transcripts of the protocol generated by the protocol encoders (at Alice and Bob), 
$(1)$ is because decoding without adversarial error recovers the message and so the number of possible encoding
 transcripts is   $\geq |{\cal M}| $,  and $(2)$ is because of the set of corrupted transcripts is larger than uncorrupted ones.

Using (\ref{eq_bound11}) and (\ref{eq_bound14}), we have,
\begin{equation}\label{eq_bound19}
\begin{split}
&\bH(M|C^{{\ell_c},r} D^{{\ell_d}})\leq N(1-\rho)\log |\Sigma|  + \delta(N\log |\Sigma|+n)+ \bH(\delta)\\
\end{split}
\end{equation}

Using (\ref{eq_bound13}), (\ref{eq_bound12}), and (\ref{eq_bound19}), gives the upper bound on $\bH(M)$,
\[
\begin{split}
\bH(M)&\leq N(1-\rho)\log |\Sigma|+2\epsilon\cdot N\log(\frac{|\Sigma|}{\epsilon})+2\epsilon n + \delta N\log |\Sigma|+\delta n +\bH(\delta) 
\end{split}
\]
The above inequality must hold for any distribution on $\cal M$, and in particular for a uniform distribution with
 $\bH(M)=\log |\cM|$. Using   $\delta\leq \bH(\delta)$ for $0\leq \delta\leq 1/2$, we have,  
\[
\frac{\log |\cM|}{N\log |\Sigma|} \leq 1-\rho+2\epsilon\cdot (1+\log_{|\Sigma|}\frac{1}{\epsilon})+ 2\epsilon n +  2\bH(\delta) +\delta n
\]

\vspace{2mm}
\noindent
{\bf Step 3.} We show that $\epsilon$-secrecy capacity of a $\rorw$-\awtppd is bounded  by,
\[
{\mathsf C}^{\epsilon} \leq 1-\rho+2\epsilon\cdot (1+\log_{|\Sigma|}\frac{1}{\epsilon})+ 2\epsilon n
\] 
Proof is by contradiction.

Let ${\mathsf C}^{\epsilon} = 1-\rho+2\epsilon\cdot (1+\log_{|\Sigma|}\frac{1}{\epsilon})+ 2\epsilon n+\hat{\xi}$,
 for some small constant $\hat{\xi}>0$. From   Definition \ref{def_smtfamily}, for any $0<\hat{\xi}'\leq \min(\frac{\hat{\xi}}{5n}, \bH^{-1}(\frac{\hat{\xi}}{5}))$, there is $N_0$, such that for any $N>N_0$,  we have  $\delta<\hat{\xi}'$ and,
\[
\begin{split}
\frac{\log |\cM|}{N\log |\Sigma|}&\geq {\mathsf C}^{\epsilon} -\hat{\xi}'\\
&= 1-\rho+2\epsilon\cdot (1+\log_{|\Sigma|}\frac{1}{\epsilon})+ 2\epsilon n+2\bH(\delta) +\delta n+\hat{\xi} -\hat{\xi}'-2\bH(\delta) -\delta n \\
&\geq 1-\rho+2\epsilon\cdot (1+\log_{|\Sigma|}\frac{1}{\epsilon})+ 2\epsilon n+2\bH(\delta) +\delta n+\hat{\xi}'\\
&> \frac{\log |\cM|}{N\log |\Sigma|}
\end{split}
\]
This 
contradicts the bound on $\frac{\log |\cM|}{N\log |\Sigma|}$, and so,
\[
{\mathsf C}^{\epsilon} \leq 1-\rho+2\epsilon\cdot (1+\log_{|\Sigma|}\frac{1}{\epsilon})+ 2\epsilon n
\]
\qed
\end{proof}

\begin{corollary} 
The perfect secrecy capacity of a  $\rorw$-\awtppd channel is bounded as,
\[
{\mathsf C}^0 \leq 1-\rho
\]
\end{corollary}


\subsection{Lower Bound on Message Round Complexity}
 
 An efficient construction of a $(0,\delta)$-\awtp code  (one message round) with rate $\mathsf{R}=1-\rho_r-\rho_w$ is given in \cite{PS14}, implying that secure transmission over \awtp channels with one message round protocols is possible if, $\rho_r+\rho_w<1$. In Section \ref{sec_ratebound}, we proved that  for \awtppd~channels, ${\mathsf C}^0 \leq 1-\rho$ and so secure communication 
with $\rho_r+\rho_w>1$ may be possible, as long as $\rho<1$.   

\begin{theorem}\label{the_rc}
Perfectly secure communication over \awtppd channel requires,\\
(i) one message round protocol, if $\rho_r+\rho_w<1$.\\
(ii) a  protocol with at least three message rounds, if $\rho_r+\rho_w\geq 1$. 
That is,
\[
\mathsf{RC}\begin{cases}
\geq 1 & \text{if  } \rho_r+\rho_w<1;\\
\geq 3 & \text{if  } \rho_r+\rho_w\geq 1.
\end{cases}
\]
\end{theorem}

We use the same  notations as in Section \ref{sec_ratebound}.

\begin{proof}
We only need to prove (ii). The protocol must have at least two message rounds and so can have one of the following forms. Note that to achieve privacy, at least one message round of \awtp channel is needed.
\end{proof}
\begin{enumerate}
\item Rnd 1: Alice $\overset{\mathsf{AWTP}}{\longrightarrow}$ Bob; Rnd  2: Alice $\overset{\mathsf{PD}}{\longrightarrow}$ Bob.
\item Rnd  1: Alice $\overset{\mathsf{AWTP}}{\longrightarrow}$ Bob; Rnd  2: Alice $\overset{\mathsf{AWTP}}{\longrightarrow}$ Bob.\item Rnd  1: Alice $\overset{\mathsf{AWTP}}{\longrightarrow}$ Bob; Rnd  2: Bob $\overset{\mathsf{PD}}{\longrightarrow}$ Alice.
\item Rnd 1: Alice $\overset{\mathsf{PD}}{\longrightarrow}$ Bob; Rnd  2: Alice $\overset{\mathsf{AWTP}}{\longrightarrow}$ Bob.
\item Rnd  1: Bob $\overset{\mathsf{PD}}{\longrightarrow}$ Alice; Rnd  2: Alice $\overset{\mathsf{AWTP}}{\longrightarrow}$ Bob.
\end{enumerate}

The third, fourth and fifth forms are not possible: in all these cases Bob's decoder will have the vector received 
through a one round \awtp channel and so the protocol cannot have rate higher than $1-\rho_r-\rho_w$.
\remove{
 The second form is equivalent to one round \awtp channel. Bob replies to Alice over \pd channel after receiving transmission from \awtp channel, which will not help him to decode the received codeword; The fourth and fifth forms are also equivalent to one round \awtp channel, since \pd channel can be used only transmit public parameters. 
}

\begin{lemma}\label{the_rbound1}
In an $(0, \delta)$-$\mathsf{AWTP_{PD}}$ protocol  of the forms (1) or (2) above,
if $\rho_r + \rho_w \geq 1$, 
then,
$$2\bH(\delta)\geq 1-\frac{1}{|\cM|}$$
\end{lemma}
Proof is in Appendix \ref{ap_the_rbound1}.


\section{ An optimal  $(0,\delta)$-\awtppd Protocol}\label{sec_construction}

We first introduce the building blocks of the \awtppd protocol, and then describe the construction.
The rate of the protocol meets the upper bound.  The protocol  has three message rounds and so meets the minimum message round complexity.
 The construction is inspired  by Shi \emph{et al}. \cite{SJST09}.

\subsection{Universal Hash Family}\label{sec_mac}


An $(N,n,m)$-hash family is a set  $\cal F$ of $N$ functions, $f: {\cal X}\rightarrow {\cal T}$,  $f\in \cal F$, where $|{\cal X}|=n$ and $|{\cal T}|=m$. 
Without loss of generality, we assume $n\geq m$.
\begin{definition}\cite{S02}
Suppose that the $(N, n, m)$-hash family $\cal F$ has range $\cal T$  which is an additive Abelian group. 
$\cal F$ is called $\epsilon$-$\Delta$ universal, if  for any two elements $x_1, x_2\in {\cal X}, x_1\neq x_2,$, and for any element $t\in {\cal T}$, there are  at most $\epsilon N$ functions $f\in \cal F$ such that $f(x_1)-f(x_2)=t$, were the operation is from the group.
\end{definition}

We will use a classic construction of $\frac{u}{q}$-universal hash family \cite{S02}. Let $q$ be a prime and $u\leq q-1$. Let the message be ${\bf x}=\{x_1, \cdots, x_{u}\}$. For $\alpha \in \bF_q$, define the universal hash function $\mathsf{hash}_\alpha$ by the rule,
\begin{equation}\label{eq_mac1}
t=\mathsf{hash}_\alpha({\bf x})=x_1\alpha+x_2\alpha^2+\cdots+x_{u}\alpha^{u} \mod q 
\end{equation}
Then $\{\mathsf{hash}_\alpha(\cdot): \alpha\in \bF_q\}$ is a $\frac{u}{q}$-$\Delta$ universal $(q, q^u, q)$-hash family.

\subsection{Randomness Extractor}\label{sec_extractor}

A randomness extractor is a function, which is applied to a weakly random entropy source (i.e.,  a non-uniform random variable), to
obtain a uniformly distributed source.

\begin{definition}\cite{DORS08}
A (seeded) $(n, m, r, \delta)$-strong extractor is a function 
$
\mathsf{Ext}: q^{n}\times q^{d}\rightarrow q^{m}
$
such that for any source $X$ with $\bH_{\infty}(X)\geq r$, we have
\[
\SD((\mathsf{Ext}(X, \mathsf{Seed}), \mathsf{Seed}), (U, \mathsf{Seed}))\leq \delta 
\]
with the seed uniformly distributed over $\bF_q^d$.

A function $\ext: q^{n}\rightarrow q^{m}$ is a (seedless) $(n, m, r, \delta)$-extractor if for any source $X$ with $\bH_{\infty}(X)\geq r$, the distribution $\ext(X)$ satisfies  $\SD(\ext(X), U)\leq \delta$.
\end{definition}

\remove{
The average min-entropy and average-case strong extractor measures the case that the adversary has side information $Y$ about source $X$. Let the average min-entropy be $\bH_{\infty}(X|Y)=-\log \mathsf{E}_{y\leftarrow Y}(\max_x\bPr(X=x|Y=y))$.
\begin{definition}\cite{DORS08}
Let the source $\bH_{\infty}(X|Y)\geq r$. A (seedless) average-case $(n, m, r, \delta)$-strong extractor is a function,
\[
{\mathsf{Ext}}: q^{n}\rightarrow q^{m}
\]
such that $\SD(({\mathsf{Ext}}(X), Y), (U, Y))\leq \delta$.
\end{definition}
}

A seedless extractor  can be constructed from Reed-Solomon (RS) codes \cite{CDS12}.
The construction works only for a  restricted class of  sources, known as {\em symbol-fixing sources}.

\begin{definition} 
An $(n, m)$ symbol-fixing source is 
a tuple of independent random variables ${\bf X} = (X_1, \cdots, X_n)$, defined over a set $\Omega$, 
such that $m$ of the variables take values uniformly and independently from $\Omega$,
and the rest have fixed values.

\end{definition}

We show a construction of a seedless $(n, m, m\log q, 0)$-extractor from RS-codes. Let $q\geq n+m$. Consider an $(n,m)$  symbol-fixing source
 ${\bf X}=(X_1, \cdots, X_n)\in \bF^n_q$ with $\bH_\infty(X)\geq m\log q$. The extraction has two steps:
\begin{enumerate}
\item 
Construct a polynomial $f(x) \in \bF_q[X]$ of degree $\leq n -1$, such that $f(i) =x_i$ for $i=0, \cdots, n-1$.

\item 
Evaluate the polynomial at $i=\{n,\cdots, n+m-1\}$. That is,
\[
\mathsf{Ext}({\bf x})=(f(n),f(n+1),\cdots ,f(n+m-1))
\]
\end{enumerate}


\subsection{$\mathsf{AWTP_{PD}}$ Protocol }\label{sec_constr}

Let  the AWTP channel have alphabet  $\Sigma=\bF_q^{u}$  where $q>2uN^2$, and the message be ${\bf m} = \{m_1,\cdots ,m_{\ell}\} \in {\cal M}$, where  $m_i\in \bF_q$. Let $N$ denote the  transmission length over the AWTP channel.
\remove{ the length of transmission on first message round over public discussion channel be $n_1$, and second message round over public discussion channel be $n_2$.} 
We use a $\frac{u}{q}$-$\Delta$ universal $(q, q^{u-1}, q)$-hash family and the seedless $(uN, \ell, \ell\log q, 0)$-extractor, above. 

\begin{center}
{\bf \awtppd Protocol }
\end{center}
\begin{framed}

\begin{itemize}
\item Rnd 1: Alice $\overset{\mathsf{AWTP}}{\longrightarrow}$ Bob. 
For  $i\in N$:\\
 Alice randomly chooses a vector ${\bf r}_i=\{r_{i,1},\cdots, r_{i,u-1}\} \in \bF_q^{u-1}$, 
and  $\beta_i\in \bF_q$. Alice sends  $c=(c_1, \cdots, c_{N}) \in \bF_q^u$ with $c_i=\{{\bf r}_i, \beta_i\}$ to Bob,
over  the  AWTP channel. \\
Bob receives $y=(y_1, \cdots, y_{N})$, where  $y_i=\{{\bf r}_i',\beta_i'\}$. 

\item Rnd 2: Bob $\overset{\mathsf{PD}}{\longrightarrow}$ Alice. \\
Bob generates random keys, $(\alpha_1, \cdots, \alpha_{N})$, $\alpha_i \in \bF_q$, for the hash family,
 and generates ${\bf t}=(t_1, \cdots, t_{N})$ where, $t_i=\mathsf{hash}_{\alpha_i}({\bf r}'_i)+\beta_i'\mod q$. 
Bob maps $d_1=\{\alpha_1, \cdots, \alpha_{N}, t_1, \cdots, t_{N}\}$ 
to a  binary vector over $\bF_2$, and sends $d_1$ to Alice, over the PD channel. Alice receives $d_1$.

\item Rnd 3: Alice $\overset{\mathsf{PD}}{\longrightarrow}$ Bob. 
\begin{itemize}
\item Alice checks,
\[ \mathsf{hash}_{\alpha_i}({\bf r}_i)+\beta_i\stackrel{?}=t_i\mod q, \; i=1\cdots N\]
and constructs a binary vector ${\bf v}=(v_1, \cdots, v_N)$, where  with $v_i=1$ if $\mathsf{hash}_{\alpha_i}({\bf r}_i)+\beta_i=t_i\mod q$, and $v_i=0$, otherwise.  

\item Let, $v_{i_1}\cdots=v_{i_s}=1$. Alice does the following.\\
--concatenates all  ${\bf r}_{i_j}$ for which $v_{i_j}=1$, and obtains $({\bf r}_{i_1}||\cdots ||{\bf r}_{i_s})$ over $\bF_q$.\\
 --uses the extractor on this string, and obtains  a uniformly random string, 
${\bf k}=\ext({\bf r}_{i_1}||\cdots ||{\bf r}_{i_s})$.
\item Alice 
encrypts the message $\bf m$ and obtains  ${\bf c}=\{c_1, \cdots, c_{\ell}\}$, where $ c_i=k_i+ m_i\mod q$ for $i=1,\cdots, \ell$.
 Alice maps $d_2=\{{\bf c}, {\bf v}\}$ (over $\bF_q$) into a binary vector  and sends it  to Bob over the PD channel. 

Bob receives $d_2$.
\end{itemize}

\item 
Bob decodes $\mathsf{Dec}(y_1, d_1, d_2)$ as follows.
\begin{itemize}
\item  Constructs the vector $({\bf r}'_{i_1}||\cdots ||{\bf r}'_{i_s})$ with  ${\bf r}'_{i_j} \in \bF_q$, 
for  all $v_{i_j}=1$ in $\bf v$. 
He uses the extractor to obtain, 
 ${\bf k}'=\ext({\bf r}'_{i_1}||\cdots ||{\bf r}'_{i_s})$. 
\item  Recovers the message ${\bf m}'$ with $m'_i=c_i- k_i'\mod q$ for $i=1,\cdots, \ell$.
\end{itemize}
\end{itemize}
\end{framed}

\vspace{3mm}

\begin{lemma}\label{le_prosec}
The \awtppd protocol above, provides perfect secrecy if $\ell \leq (u-1)(1-\rho)N$.
\end{lemma}

\begin{lemma}\label{le_prorel}
The probability of decoding error in the \awtppd protocol   is $\delta\leq \frac{uN}{q}$.
\end{lemma}


\begin{lemma}
The rate of the \awtppd  protocol family is  $\mathsf{R}=1-\rho$.
\end{lemma}

\begin{proof}
For a  small $\xi>0$,  let the parameters of \awtppd protocol be chosen as $u=\frac{1}{\xi}$, $q>2uN^2$, $\ell=(u-1)(1-\rho)N$, $N_0\geq \frac{1}{\xi}$ and $\Sigma=\bF_q^{u}$.  For uniform  message  distribution, we have $\log |\cM|=\ell \log q$, and 
so for any $N>N_0$, the rate of \awtppd protocol family  is given by,
\[
\frac{\log |{\cal M}|}{N\log |\Sigma|}=\frac{(u-1)(1-\rho)N\log q}{uN\log q} =(1-{\xi})(1-\rho)
\geq 1-\rho-\xi
\]
\remove{
\[
\begin{split}
\frac{\log |{\cal M}|}{N\log |\Sigma|}&=\frac{(u-1)(1-\rho)N\log q}{uN\log q}\\
&=(1-{\xi})(1-\rho)\\
&\geq 1-\rho-\xi
\end{split}
\]
}
The probability of decoding error is bounded by,
\[
\delta \leq \frac{uN}{q}\leq \frac{1}{2N}\leq  \frac{\xi}{2}\leq \xi
\]
\end{proof}

\begin{theorem}
For any small $\xi>0$, the protocol above  is a $(0, \delta)$-\awtppd protocol with rate $\mathsf{R}(\Pi^N)=1-\rho-\xi$.
The transmission alphabet over the \awtp channel is of size $|\Sigma|=q^{\frac{1}{\xi}}$, and the decoding error is $\delta < \xi$. 
The rate of  the protocol 
 approaches $\mathsf{R}=1-\rho$ as, $N\rightarrow \infty$. 
 The protocol has \rc=3 and the decoder computation is $\mathcal{O}((N\log q)^2)$. 
\end{theorem}


\section{\awtppd Protocol and SMT-PD}\label{sec_smtawtp}


In  {\em SMT-PD }
 a  sender $\cS$ (Alice) and a receiver $\cR$ (Bob) interact over $N$ node disjoint paths ({\em wires}) in  a synchronous network and  a public discussion channel. \emph{Wires and the PD  both are used for two-way communication}. 
An SMT-PD protocol proceeds in rounds. \emph{In each round, Alice (Bob) sends
protocol messages  over wires and/or the PD channel, which will be received by Bob (Alice) before the end of the  round.}
(Note that a round in SMT-PD may consist of one or  two message rounds.)
A computationally unbounded adversary (Eve) can corrupt up to $t$ wires. Eve can 
eavesdrop, modify or block messages sent over a corrupted wire. 
Adversary is adaptive and can corrupt wires any time during the protocol execution and after observing communications over the wires that she has corrupted so far. We consider protocol families ${\bf \Pi}= \{ \Pi^N: N \in \mathbb{N} \}$ defined for   $t= \rho N$  where $0<\rho<1$ is a constant.

\begin{definition}\label{def_SMT}
A protocol between $\cS$  and $\cR$  is an $(\epsilon,\delta)$-\emph{secure message transmission with public discussion} ($(\epsilon,\delta)$-SMT-PD) protocol if 
 the following two conditions are  satisfied.

\begin{itemize}
\item Privacy: For every two messages $m_1, m_2\in \cM$ and randomness $r_E$ used by Eve, 
\[
\begin{split}
\max_{m_1, m_2}\SD(&\mathsf{View}_{\mathsf E}(\mathsf{SMT_{PD}}(m_1),r_E),\mathsf{View}_{\mathsf E}(\mathsf{SMT_{PD}}(m_2), r_E))\leq \epsilon,
\end{split}
\]
where the probability is  over the randomness of $\cS, \cR$.

\item Reliability: For any  message $M_\cS$ chosen by Alice, Bob recovers the message with probability larger than $1-\delta$; that is, 
\[
\bPr(M_\cR\neq M_\cS)\leq \delta,
\]
where the probability is over the randomness of players $\cS, \cR$ and Eve.
\end{itemize}
\end{definition}

\begin{remark}
In the above definition of SMT-PD, (i) $S_r= S_w$, and for $|S_r|=|S_w|= \rho N$, (ii) wires are used for two-way communication, and (iii) 
in each   message round of the protocol, Alice (Bob) can invoke both types of channels  simultaneously (wires and the PD) and so send two protocol message.
In our model  in Section \ref{sec_awtppd} however, (i)  $S_r$ and $S_w$ can be chosen
arbitrarily, (ii)  \awtp is from Alice to Bob only, and 
(iii) in each message round one message  over  one channel (\awtp, or PD) can be sent.
\end{remark}

Efficiency parameters  of an SMT-PD protocol are, \emph{Round Complexity} $\mathsf{RC}$, \emph{Transmission Rate} $\mathsf{TR}$, and \emph{computational complexity}. 
\begin{itemize}
\item $\mathsf{RC}$ is the number of rounds  of a protocol.   We also use \rc~  to denote  message round complexity of these protocols. 
\item  $\mathsf{TR}$ is the number of communicated bits for transmitting a single message bit.
Let ${\cal W}_i$ denote the set of possible transmissions on wire $i$. The transmission rate of an SMT-PD protocol is given by,
\[
\mathsf{TR}=\frac{\sum_{i=1}^N \log |{\cal W}_i|}{\log |\cM|}
\]
An SMT-PD protocol is {\em optimal} if the transmission rate is of  the order (Big $\mathcal{O}$ notation) of  the lower bound.
\item An SMT-PD protocol is computationally efficient if the computational complexity of the sender and the receiver algorithms,
is polynomial in $N$. 
\end{itemize}

\subsection{  \awtppd and One-way SMT-PD}



AWTP codes are defined over an alphabet $\Sigma$ and all components of a codeword are elements of $\Sigma$.
In SMT protocols however, the set of transmissions over different wires may be different.

\begin{definition}[Symmetric SMT]
An SMT protocol is called a {\em symmetric} if the protocol remains invariant under any permutation of  the wires.
\end{definition}

Let ${\cal W}^i_j, j=1\cdots N, i=1\cdots r$, 
 denote the set of possible transmissions on wire $j$  in an $r$-round SMT protocol.
For a symmetric protocol, ${\cal W}^i_j = {\cal W}^i$  is independent of $j$.
{\em All known constructions of  threshold  SMT protocols are symmetric.}

\begin{definition}
A {\em one-way symmetric secure message transmission with public discussion (\smtone) protocol}  is an
 SMT-PD protocol in which transmission over wires is in one direction (from  Alice to Bob, or  Bob to Alice). The protocol is invariant under any permutation of  the wires. 
 The $N$ wires and the \pd channel, can be invoked  simultaneously. 
\end{definition}

We consider protocols where  Alice wants to send
 a message to Bob and so \awtp channel is used by Alice.

\begin{theorem}\label{the_relation1}
There is a one-to-one correspondence between restricted $(\epsilon,\delta)$-\awtppd protocols and \smtone protocols. 
The following results on the latter protocols, follow from the results  on the  former  in Section \ref{sec_upperbound}. 
\begin{enumerate}

\item The lower bound on the transmission rate of  a \smtone protocol is,
\begin{equation}\label{epsec}
\begin{split}
&\mathsf{TR}\geq \frac{N}{N-t+\epsilon'+2\bH(\delta)N+\delta nN}
\end{split}
\end{equation}
where $\epsilon'=2N\epsilon (1+\log_{|\cal W|}\frac{1}{\epsilon})+2\epsilon nN$.

\noindent
For protocols with perfect secrecy ($\epsilon=0$) we have, 
\begin{eqnarray}\label{perf}
\mathsf{TR}\geq \frac{N}{N-t+2\bH(\delta)N+\delta nN}.
\end{eqnarray}
\noindent

\item The lower bound on the  message round complexity of a $(\epsilon,\delta)$-$\mbox{SMT}^{[ow]}$-$\mbox{PD}$ protocol is  three. 
\end{enumerate}
\end{theorem}
\begin{proof}
It is easy to see that an \smtone protocol gives a \awtppd protocol: 
using the same conversion as in \cite{PS13}  a protocol message over \smtone wires gives a  protocol message over \awtp channel by considering   wire $i$ as component $i$ of the \awtp codeword;
 messages over \pd will stay the same in both. The conversion holds in reverse direction also. 
The lower bound on transmission rate  follows by noting that the 
 transmission rate of  a  \smtone protocol is the inverse of the rate of the corresponding \awtppd protocol, and
 so the upper bound on the rate of \awtppd protocols implies a lower bound on the transmission rate of  \smtone protocols. 
 The lower bound on  message round complexity follows from the similar  bound on the corresponding  \awtppd protocols.
 Details are given in Appendix \ref{ap_the_relation1}. 
\end{proof}

\subsubsection{Construction} \label{const}
A $(\epsilon,\delta)$-\awtppd protocol gives a 
restricted-$(\epsilon,\delta)$-\awtppd protocol with $\rho=\rho_r = \rho_w$. This latter, 
using the protocol  conversion in Theorem \ref{the_relation1}, gives an  \smtone protocol.  
In Section \ref{sec_constr} we gave the construction of a $(0,\delta)$-\awtppd protocol with minimum number of  message rounds and rate approaching the capacity of the $(\rho_r, \rho_w)$-\awtp channel.
This  leads to the following.

\begin{lemma}\label{le_constrsmt1}
There is a three  message  round \smtone protocol,
with  transmission rate,  $\mathcal{O}(\frac{N}{N-t})$,  and  decoding computational complexity equal to, $\mathcal{O}((N\log q)^2)$.
\end{lemma}

\subsubsection{Comparison with known results}

In \cite{GGO10} it was shown that secure SMT-PD protocols exist 
for $N \geq t+1$, and the following lower bound on the transmission  rate  was derived,
\begin{eqnarray} \label{smt-pd}
\mathsf{TR}\geq \frac{N\cdot (-\log (\frac{1}{|\cM|} +2\epsilon)-\bH(\sqrt{\delta})-2m\sqrt{\delta})}{(N-t)m}.
\end{eqnarray}
Here, $m=\log |\cM|$. 
The bound gives a lower bound on the transmission rate of \smtone protocols as an \smtone protocol is 
an SMT-PD protocol with extra restriction. 
 None of the two bounds, (\ref{epsec}) and (\ref{smt-pd}), completely dominates the other: 
\begin{enumerate}
\item For $\epsilon=0$ and $\delta>0$, (\ref{smt-pd}) will be  a  tighter bound. 
This is because  for perfectly secure SMT-PD,  for 
 $\log |\cM|\gg\bH(\sqrt{\delta})$, the bound  
 (\ref{smt-pd}) can be written as, 
\begin{eqnarray}  
\mathsf{TR}\geq \frac{N}{N-t}\frac{(1-2\sqrt{\delta})\log |\cM| }{\log |\cM|}.
\end{eqnarray}

From,
\[
\frac{N}{N-t}(1-2\sqrt{\delta})
= \frac{N-2\sqrt{\delta}N}{N-t}
\geq \frac{N}{N-t+2\sqrt{\delta}N}
\geq \frac{N}{N-t+2\bH(\delta)N+\delta n N},
\]

we conclude that the bound  (\ref{smt-pd})  
is tighter than the bound   Eq. (\ref{perf}).

\item For $\delta \approx 0$ and $\epsilon=\frac{a}{|\cM|}$ however, (\ref{epsec}) could give a higher value. For example, consider $|\cM|=2^N$, $\epsilon=\frac{1}{|\cM|}$, and $n={\cO} (N)$. The bound  (\ref{smt-pd}) is, 
\[
\mathsf{TR} \geq \frac{N\cdot (-\log (\frac{1}{|\cM|} +\frac{2}{|\cM|}))}{(N-t)\log |\cM|}
= \frac{N}{N-t}(1-\frac{\log 3}{N}), 
\] 
and the bound (\ref{epsec}) is,
\[
\mathsf{TR} \geq \frac{N}{N-t+\epsilon'}
 \geq \frac{N}{N-t+2\frac{N}{2^N} (1+N)+\cO(\frac{N^2}{2^N})} 
 =\frac{N}{N-t+\cO(\frac{N^2}{2^N})}. 
\]
Hence the bound   (\ref{epsec}) is tighter than (\ref{smt-pd}) for large $N$ approaching  infinity. 
\end{enumerate}



In \cite{SJST09}, it was shown that the minimum round complexity of an SMT-PD protocol is three,   and \pd must be invoked in at least two rounds.  Since an \smtone 
 is an SMT-PD with extra restrictions, 
 the same bounds also hold for them.
The rate-optimal \smtone protocol in Section \ref{sec_constr} has three  message rounds, two of which use
\pd,  and so achieves the lower bound on the number of rounds of 
\smtone protocols. 
\vspace{3mm}

\begin{table}[H]\label{state100}

\begin{center}
\centering \caption{Comparison with SMT-PD protocols  }	\renewcommand{\arraystretch}{1.5}
\renewcommand{\multirowsetup}{\centering}
\begin{tabular}{|c|c|c|c|c|c|}	
\hline \textbf{SMT-PD} & \pbox{30mm}{\bf Num of Message\\ Rnds} &  \pbox{30mm}{\bf Read  and\\ Write Sets} & \pbox{30mm}{\bf Communication\\ over $\sfPD$}  & \pbox{25mm}{\bf Info. \\ Rate} &  
 {\bf Trans. Rate}  \\
\hline Shi \textit{et al.} \cite{SJST09}  &  1 $\sfSMT$ 2 $\sfPD$ & $S_r=S_w$ $\rho\leq 1$ & $\log |\cM|$ & $1-\frac{t}{N}-\xi$
&   $\cO(\frac{N}{N-t})$ \\ 
\hline Garay \textit{et al.} Prot.  I  \cite{GGO10}  &  1 $\sfSMT$ 2 $\sfPD$ & $S_r=S_w$ $\rho\leq 1$ & $\log |\cM|$ & $1-\frac{t}{N}-\xi$ 
 & $\cO(\frac{N}{N-t})$ \\ 
\hline Garay  \textit{et al.}  Prot. II \cite{GGO10} & 2 $\sfSMT$ 2 $\sfPD$ & $S_r=S_w$ $\rho\leq 1$ & $\log\log |\cM|$ & $c(1-\frac{t}{N})$ 
 &   $\cO(\frac{N}{N-t})$ \\ 
\hline This Work  &  1 $\sfSMT$ 2 $\sfPD$ & $\rho\leq 1$ & $\log |\cM|$ & $1-\frac{t}{N}-\xi$
& $\cO(\frac{N}{N-t})$ \\ 
\hline								
\end{tabular}
\end{center}

\end{table}

$c$ is a constant which is no more than $\frac{1}{3}$. The information rate of Protocols I and II are derived  in  Appendix \ref{ap_sec6}. 

\section{Conclusion}
We motivated and  introduced  $\mathsf{AWTP_{PD}}$, where Alice and Bob, in addition to the 
\awtp channel, have access to a public discussion channel and showed that with this new
resource, secure  communication is possible even when $\rho_r+\rho_w\geq 1$ as long as $\rho<1$.  We derived an upper  bound on the information rate, and a lower bound on the number of message rounds of  protocols that provide $\epsilon$-secrecy and $\delta$-reliability, and constructed an optimal protocol family that achieve both these bounds.
We showed the relationship between \awtppd and \smtone protocols in which wires are used by Alice only, and gave the construction of an optimal  \smtone protocol with minimum number of message rounds. A three-round protocol SMT-PD (two-way wires) with the same rate had been constructed in \cite{SJST09}.  
Our construction shows that 
 assuming one-way communication over  wires does not affect the number 
 of message rounds of the optimal protocols.
 
 \smtone protocols  remove the restriction of $\rho_r+\rho_w\leq 1$
and allow secure communication when $\rho_r+\rho_w\geq 1$  as long as $|S_r\cup S_w| < N$.
In our model although we allow interaction, but  the \awtp channel is  one-way.  
An interesting open question is to obtain rate and \rc~ lower bounds for the case that interaction over the \awtp channel is possible.

\bibliographystyle{abbrv}
\bibliography{1.bib}

\begin{appendix}


\section{Proof of Section \ref{sec_upperbound}}

\subsection{Proof of Lemma \ref{le_smt_bound3}}\label{ap_smt_bound3}

\begin{proof}
The proof is similar to Theorem 4.9 \cite{BTV12} and uses Pinsker's Lemma:
\begin{lemma}
Let $P$, $Q$ be probability distributions. Let $ \SD(P , Q)\leq \epsilon$. Then 
\[
\bH(P ) - \bH(Q) \leq 2\epsilon \cdot \log(\frac{|P\cup Q|}{\epsilon})
\]
\end{lemma}

Let  the random variable of 
the adversarial view,   $V_E$,  be over the set 
 $V_E$. 
According to the definition of $\epsilon$-secrecy (Definition \ref{def_smt}), for any pair of message $m_1, m_2\in \cM$, the statistical distance between the distribution of  $V_E$ when Alice sends $m_1$, and the distribution of $V_E$ when Alice sends $m_2$, is no more than $\epsilon$. That is
\[
\begin{split}
\epsilon&\geq \max_{m_1,m_2}\SD(V_E|M=m_1, V_E|M=m_2)\\
&\geq \max_{m_1, m_2}\sum_{v\in \mathcal{V}_E}|\bPr(v|m_1)-\bPr(v|m_2)|
\end{split}
\]

Assuming distribution $\Pr(m)$ on $\cal M$, this implies, 
\begin{equation}\label{eq_smt_bound1}
\begin{split}
&\SD(V_E, V_E| M=m)\\
&=\frac{1}{2}\sum_{v\in \mathcal{V}_E}|\bPr(v|m)-\bPr(v) |\\
&=\frac{1}{2}\sum_{v\in \mathcal{V}_E}|\bPr(v|m)-\sum_{m'}\bPr(v|m')\bPr(m') |\\
&=\frac{1}{2}\sum_{v\in \mathcal{V}_E}|\sum_{m'}\bPr(m')(\bPr(v|m)-\bPr(v|m')) |\\
&\leq \frac{1}{2}\sum_{v\in \mathcal{V}_E}\sum_{m'}\bPr(m')|\bPr(v|m)-\bPr(v|m')|\\
&=\sum_{m'}\bPr(m')\frac{1}{2}\sum_{v\in \mathcal{V}_E}|\bPr(v|m)-\bPr(v|m')|\\
&\leq \sum_{m'}\bPr(m')\max_{m_1,m_2}\SD(V_E|M=m_1, V_E|M=m_2) \\
&\leq \epsilon
\end{split}
\end{equation}

From  Pinsker Lemma  and Eq. (\ref{eq_smt_bound1}), we have,  
\[
\begin{split}
&\bH(V_E)-\bH(V_E|M=m)\leq 2\epsilon\cdot \log(\frac{|{\cal V}_E|}{\epsilon})
\end{split}
\]
From $|{\cal V}_E|\leq 2^{n}\times |\Sigma|^{N}$, it implies, 
\[
\bH(V_E)-\bH(V_E|M=m)\leq 2\epsilon\cdot \log(\frac{|\Sigma|^{N}}{\epsilon})+2\epsilon n
\]

So the difference between $\bH(M)$ and $\bH(M|V_E)$ is 
\begin{equation}
\begin{split}
\bH(M)-\bH(M|V_E)&=\bH(V_E)-\bH(V_E|M)\\
&=\bH(V_E)-\sum_{m\in \cM}\bPr(m)\bH(V_E|m)\\
&=\sum_{m\in \cM}\bPr(m) (\bH(V_E)-\bH(V_E|m))\\
&\leq 2\epsilon N\cdot \log(\frac{|\Sigma|}{\epsilon})+2\epsilon n
\end{split}
\end{equation}
\end{proof}

\subsection{Proof of Lemma  \ref{le_smt_bound2}}\label{ap_smt_bound2}

\begin{proof}
Let $\delta'=\bH(\delta)+\delta \log |\cM|$. The proof has two steps. 

\begin{enumerate}
\item 
We show that $\bH(M|C^{{\ell_c},a}Y^{{\ell_c},w} C^{{\ell_c},d} D^{\ell_d})\leq \delta'$. 

Let $\delta = \Pr(M_\cR \neq M_\cS)$. From Fano's inequality,
\[
\begin{split}
&\bH(\delta ) + \delta \log |\cM| \geq \bH(M_\cS|M_\cR) \geq \bH(M_\cS | Y^{\ell_c} D^{\ell_d})
\end{split}
\]

Here $\{y^{\ell_c}, d^{\ell_d}\}$, is the received vectors of Bob. 
Since $y^{{\ell_c}}=\{c^{{\ell_c},a},y^{{\ell_c},w},c^{{\ell_c},d}\}$, we have,
\begin{equation}\label{eq_dd_3}
\begin{split}
\bH(M_\cS | C^{{\ell_c},a} Y^{{\ell_c},w} C^{{\ell_c}, d} D^{\ell_d})\leq \bH(M_\cS|M_\cR)\leq \delta'
\end{split}
\end{equation}

\item 
We show that 
\[
\begin{split}
&\bH(M_\cS|C^{{\ell_c},a} C^{{\ell_c},d} D^{{\ell_d}})\leq \delta'+\bI(Y^{{\ell_c},w}; C^{{\ell_c},w}|C^{{\ell_c},a} C^{{\ell_c},d} D^{{\ell_d}})
\end{split}
\]
Writing the conditional entropy in two ways, we have,
\[
\begin{split}
&\bH(M_\cS Y^{{\ell_c},w} | C^{{\ell_c},a} C^{{\ell_c},d} D^{{\ell_d}})\\
&=\bH(M_\cS|C^{{\ell_c},a} Y^{{\ell_c},w} C^{{\ell_c},d} D^{{\ell_d}})+\bH(Y^{{\ell_c},w}|C^{{\ell_c},a} C^{{\ell_c},d} D^{{\ell_d}})\\
&=\bH(M_\cS|C^{{\ell_c},a} C^{{\ell_c},d} D^{{\ell_d}})+\bH(Y^{{\ell_c},w}|C^{{\ell_c},a} C^{{\ell_c},d} D^{{\ell_d}} M_\cS)
\end{split}
\]
and so,
\begin{equation}\label{eq_dd_1}
\begin{split}
&\bH(M_\cS|C^{{\ell_c},a} C^{{\ell_c},d} D^{{\ell_d}})\\
&=\bH(M_\cS|C^{{\ell_c},a}Y^{{\ell_c},w} C^{{\ell_c},d} D^{{\ell_d}})+\bH(Y^{{\ell_c},w}|C^{{\ell_c},a} C^{{\ell_c},d} D^{{\ell_d}})-\bH(Y^{{\ell_c},w}|C^{{\ell_c},a} C^{{\ell_c},d} D^{{\ell_d}} M_\cS)\\
\end{split}
\end{equation}

Because of the Markov chain $M_\cS \rightarrow C^{{\ell_c}}D^{{\ell_d}}(=C^{{\ell_c},a}C^{{\ell_c},w}C^{{\ell_c},d} D^{{\ell_d}})
\rightarrow C^{{\ell_c},w}$,
we have
\begin{equation}\label{eq_dd_2}
\begin{split}
&\bH(Y^{{\ell_c},w}|C^{{\ell_c},a} C^{{\ell_c},d} D^{{\ell_d}} M_\cS)\geq \bH(Y^{{\ell_c},w}|C^{{\ell_c},a}C^{{\ell_c},w}C^{{\ell_c},d} D^{{\ell_d}})
\end{split}
\end{equation}

From (\ref{eq_dd_3}) (\ref{eq_dd_1}) and (\ref{eq_dd_2}), we have,
\begin{equation}\label{eq_ap_upbound2}
\begin{split}
&\bH(M_\cS|C^{{\ell_c},a} C^{{\ell_c},d} D^{{\ell_d}})\\
&=\bH(M_\cS|C^{{\ell_c},a}Y^{{\ell_c},w} C^{{\ell_c},d} D^{{\ell_d}})+\bH(Y^{{\ell_c},w}|C^{{\ell_c},a} C^{{\ell_c},d} D^{{\ell_d}})-\bH(Y^{{\ell_c},w}|C^{{\ell_c},a} C^{{\ell_c},d} D^{{\ell_d}} M_\cS)\\
&\leq \delta'+\bH(Y^{{\ell_c},w}|C^{{\ell_c},a} C^{{\ell_c},d} D^{{\ell_d}})-\bH(Y^{{\ell_c},w}|C^{{\ell_c},a}C^{{\ell_c}, w} C^{\ell,d} D^{{\ell_d}})\\
&\leq \delta'+\bI(Y^{{\ell_c},w}; C^{{\ell_c},w}|C^{{\ell_c},a} C^{{\ell_c},d} D^{{\ell_d}})
\end{split}
\end{equation}

Note that $Y^{{\ell_c},w}=C^{{\ell_c},w}+E^{{\ell_c},w}$ where $E^{{\ell_c},w}$ is a uniformly distributed variable, and 
so \begin{equation}\label{eq_ap_upbound3}
\bI(Y^{{\ell_c},w}; C^{{\ell_c},w}|C^{{\ell_c},a} C^{{\ell_c},d} D^{{\ell_d}})=0
\end{equation}

This means that,
$$\bH(M_\cS|C^{{\ell_c},a} C^{{\ell_c},d} D^{{\ell_d}}) \leq \delta'$$

\end{enumerate}
\end{proof}

\subsection{Proof of Lemma \ref{the_rbound1}}\label{ap_the_rbound1}

\begin{proof}
We only show that it is impossible to have a two message round $(0, \delta)$-\awtppd protocol 
of form with rate higher than  $1 -\rho_r-\rho_w$:

\begin{enumerate}
\item Rnd 1: Alice $\overset{\mathsf{AWTP}}{\longrightarrow}$ Bob

\item Rnd 2: Alice $\overset{\mathsf{PD}}{\longrightarrow}$ Bob
\end{enumerate}

The impossible result to have a two message round  $(0, \delta)$-\awtppd protocol 
of form: Rnd 1, Alice $\overset{\mathsf{AWTP}}{\longrightarrow}$ Bob; Rnd 2, Alice $\overset{\mathsf{AWTP}}{\longrightarrow}$ Bob, with rate higher than  $1 -\rho_r-\rho_w$, can be proved similarly. 

We only consider the case that $\rho_r=1-\rho_w$. The case that $\rho_r>1-\rho_w$ can be proved  similarly. 

We consider  a pair of adversaries, $\{\mathsf{Adv}_2, \hat{\mathsf{Adv}}_2\} $, both with the following properties:
\remove{
a group of adversary $\cA_2$. Each adversary in $\cA_2$ explores a specific adversarial strategy to against the protocol. 
}
\begin{enumerate}
\item Adversary selects the reading and writing sets 
  before the start of the \awtppd protocol.
\item 
 Adversary also chooses the error $e^w$ 
 randomly and uniformly from $\Sigma^{\rho_wN}$. That is $\bPr(e^w)=\frac{1}{|\Sigma^{\rho_wN}|}$.
\end{enumerate}

Adversary  $\mathsf{Adv}_2$ uses the read and write sets, $S^{r}=\{S^{a}, S^{b}\}$ and $S^{ w}=\{S^b, S^c\}$.

Because of $\rho_r=1-\rho_w$, we have $[N]=S^aS^bS^cS^d$ and  $|S^{b}|= |S^{d}|$
 
Adversary  $\hat{\mathsf{Adv}_2 }$ uses the read and write sets, $\hat{S}^{ r}=\{{S}^{ a}, {S}^{d}\}$,
 and  $\hat{S}^{w}=\{{S}^{c}, {S}^{d}\}$.

We have the following:
\begin{itemize}
\item 
Since the reading and writing capabilities  of adversary $\mathsf{Adv}_2$ is same as the adversary $\mathsf{Adv}_1$  in Section \ref{sec_upperbound}, using Lemma \ref{le_smt_bound2} we have, 
\begin{equation}\label{eq_sw1}
\bH(M|C^{a}C^{d}D)\leq \bH(\delta)+\delta(\bH(M)-1)
\end{equation}

\item Since the reading capability of  $\hat{\mathsf{Adv}}_2$ is the same as $\mathsf{Adv}_1$  in Section \ref{sec_upperbound}, from Lemma \ref{le_smt_bound1}, we have, 
\begin{equation}\label{eq_sw2}
\bI({M}; {C}^{a} {C}^{d} {D})= 0
\end{equation}

\item From (\ref{eq_sw1}) (\ref{eq_sw2}), we obtain, 
\[
\begin{split}
\bH(\delta)+\delta\bH(M)&\geq \bH(M|C^{\ell,a}C^{\ell,d}D^{\ell})\geq \bH(M)
\end{split}
\]
and so, 
\[
\frac{\bH(\delta)}{1-\delta}\geq \bH(M)
\]
Since $0\leq \delta<\frac{1}{2}$ and the message is uniformly distributed, we have,
\[
1-2\bH(\delta)\leq 2^{-2\bH(\delta)}\leq 2^{-\bH(M)}=\frac{1}{|\cM|}
\]
and,
$2\bH(\delta)\geq 1-\frac{1}{|\cM|}$.
\end{itemize}

\end{proof}


\section{Proof of Section \ref{sec_construction}}

\subsection{Proof of Lemma \ref{le_prosec}}\label{ap_le_prosec}

\begin{proof}
First, assume the adversary reads the last $\rho_r N$ components of $c$, and the first $(1-\rho)N$ components
 is the set of components that  is neither read,  nor written to, by the adversary. Let $v_E'=\{{\bf r}_{(1-\rho_r)N+1}\cdots {\bf r}_{N}, \beta_{(1-\rho_r)N+1}\cdots \beta_N, {\alpha}_1\cdots \alpha_N,  t_1\cdots t_N, v_0\cdots v_N\}$ denote the view of the adversary,  except for $\bf c$. 

If $\ell \leq (u-1)(1-\rho)N$, the vector of random variables, $({\bf r}_{i_1}||\cdots||{\bf r}_{i_s})$, corresponds to a symbol-fixing source. The components that the adversary do not read are uniformly distributed and are independent from the adversary's view $v_E'$, and the components  that the adversary reads are  determined and fixed. So the randomness $\bf k$ that is  generated from the extractor, is uniformly distributed and is independent of the adversarial view. That is,
\begin{equation}\label{eq_sssi_1}
\bPr({\bf k} | v_E')=\bPr({\bf k})
\end{equation}

Second, since Alice selects the message ${\bf m}\in \cM$  independent from ${\bf k}$ and $v_E'$,  we have 
$\bPr({\bf m} | {\bf k}, v_E')=\bPr({\bf m})$. 
For any message ${\bf m}\in \cM$, we have,
$$\bPr({\bf m})\leq \bPr({\bf m} | v_E')\leq \bPr({\bf m} | {\bf k},v_E')=\bPr({\bf m})$$
This implies,
\begin{equation}\label{eq_sssi_3}
\bPr({\bf m})= \bPr({\bf m}|v_E')= \bPr({\bf m} | {\bf k}, v_E')
\end{equation}
and so we have,
\begin{equation}\label{eq_sssi_2}
\begin{split}
\bPr({\bf k}|{\bf m}, v_E')&=\frac{\bPr({\bf k}, {\bf m}, v_E')}{\bPr({\bf m}, v_E')}\\
&=\frac{\bPr({\bf m} |{\bf k},  v_E')\bPr({\bf k}, v_E')}{\bPr({\bf m} | v_E')\bPr(v_3E')}\\
&=\bPr({\bf k} | v_E')
\end{split}
\end{equation}

Third, the adversarial view for any ${\bf m}\in \cM$ is $v_E=\{{\bf c}, v_E'\}$, and so,
\[
\begin{split}
\bPr(v_E|{\bf m})&=\bPr({\bf c} ,v'_E|{\bf m})\\
&=\bPr({\bf c}|{\bf m}, v'_E)\bPr(v'_E|{\bf m})\\
&\overset{(1)}{=}\bPr({\bf k}|{\bf m},v'_E)\bPr(v'_E)\\
&\overset{(2)}{=}\bPr({\bf k})\bPr(v'_E)\\
\end{split}
\]
where, (1) is from $c_i=k_i+m_i\mod q$ for $i=1\cdots \ell$, and 
(2) is from (\ref{eq_sssi_1}) and (\ref{eq_sssi_2}).

This means the statistical distance between adversarial views of any two messages ${\bf m}_1, {\bf m}_2\in \cM$, 
 is zero and the \awtppd protocol is perfectly secure. That is,
\[
\begin{split}
&\SD(\mathsf{View}_E|{\bf m}_1, \mathsf{View}_E|{\bf m}_2)=\sum_{v_E\in \mathsf{View}_E}|\bPr(v_E|{\bf m}_1)-\bPr(v_E|{\bf m}_2)|=0
\end{split}
\]
\end{proof}

\subsection{Proof of Lemma \ref{le_prorel}}\label{ap_le_prorel}

\begin{proof}
First, we show the probability that vector $({\bf r}_{i_1}, \cdots, {\bf r}_{i_s})\neq ({\bf r}'_{i_1}, \cdots, {\bf r}'_{i_s})$ is no more than $\frac{uN}{q}$.
This is from,  
\begin{equation}\label{eq_sads_1}
\begin{split}
&\bPr(({\bf r}_{i_1}, \cdots, {\bf r}_{i_s})\neq ({\bf r}'_{i_1}, \cdots, {\bf r}'_{i_s}))\\
&\leq \sum_{i=1}^N\bPr({\bf r}_i\neq {\bf r}'_i)\\
&= \sum_{i=1}^N\bPr({{\bf r}}_i\neq {{\bf r}}'_i, v_i=1)\\
&\leq \sum_{i=1}^N\bPr({{\bf r}}_i\neq {{\bf r}}'_i,[\mathsf{hash}_{\alpha_i}({\bf r}_i)-\mathsf{hash}_{\alpha_i}({\bf r}'_i)]=[\beta_i'-\beta_i])\\
&\leq \frac{uN}{q}
\end{split}
\end{equation}
Second, for the two random vectors ${\bf k}=\mathsf{Ext}({\bf r}_{i_1}, \cdots, {\bf r}_{i_s})$ and ${\bf k}'=\mathsf{Ext}({\bf r}'_{i_1}, \cdots, {\bf r}'_{i_s})$, we have, 
\begin{equation}\label{eq_sads_2}
\bPr({\bf k}\neq {\bf k}')\leq \bPr(({\bf r}_{i_1}, \cdots, {\bf r}_{i_s})\neq ({\bf r}'_{i_1}, \cdots, {\bf r}'_{i_s}))
\end{equation}

Third, Bob correctly  receives 
$d_2=\{{\bf c},{\bf v}\}$ sent by Alice and so, $m_i+k_i=m_i'+k_i'\mod q$ for $i=1\cdots \ell$. 
That is, the probability that the message ${\bf m}\neq {\bf m}'$, is the same as the probability 
 ${\bf k}\neq {\bf k}'$.  That is, 
\begin{equation}\label{eq_sads_3}
\bPr({\bf m}\neq {\bf m}')= \bPr({\bf k}\neq {\bf k}')
\end{equation}
From (\ref{eq_sads_1}) (\ref{eq_sads_2}) (\ref{eq_sads_3}), there is $\bPr({\bf m}\neq {\bf m}')= \bPr({\bf k}\neq {\bf k}')\leq \frac{uN}{q}$. 
\end{proof}


\section{Proof of Section \ref{sec_smtawtp}}

\subsection{Proof of Lemma \ref{the_relation1}}\label{ap_the_relation1} 

\begin{proof}
First, we show that there is a one-to-one correspondence between \smtone  protocols and 
 restricted $(\epsilon,\delta)$-\awtppd protocols,  in the sense that  given one of the former, a corresponding 
one in the latter can be constructed, and vice versa, and (ii) given one of the that the security and reliability parameters of the two protocols are
 the same.
\begin{enumerate}

\item Consider a \smtone protocol, with a fixed public numbering of wires.
Recall that the in each message round of the \smtone protocol, both the wires and the \pd can be invoked by Alice,
while in our \awtppd model, only one type channel is invoked by Alice in each message round.
In  both models Bob can invoke the \pd in each message round.
We can convert the protocol messages in message round $i$ of a \smtone protocol
to the  protocol messages of message round $j$ and $j+1$, of a \awtppd protocol.
In message round $i$, transmissions over wire 1 to  $N$, defines a codeword of length 
$N$ in the $i^{th}$ message round $j$ of the  \awtp. 
The transmission over the \pd directly defines the transmission over the \pd in \awtppd,  in the $j+1$ message round.
 Each message round of the transmission over \pd, when invoked by Bob in the \smtone,  defines  a transmission over the \pd for the a\smtone protocol. 
The above transformation gives a \awtppd from a \smtone.  Similarly, a \awtppd protocol defines an \smtone protocol.
 
So a restricted $(\epsilon,\delta)$-\awtppd protocol can be constructed from \smtone protocol. Similarly, a \smtone protocol can also be constructed from restricted $(\epsilon,\delta)$-\awtppd protocol.

\item  \awtppd and \smtone definitions of secrecy and reliability are the  same.  Definition of $\epsilon$-secrecy in both primitives requires  
statistical distance of the adversary's view for two messages chosen by the adversary (Compare definition \ref{def_SMT} and definition \ref{def_smt}), to be bounded by $\epsilon$.  For $\delta$-reliability,  both primitives require the 
probability of outputting the correct message to be at least $1-\delta$, and the probability of outputting the wrong message  to be at most  $\delta$. 
 \remove{ Adversary's capabilities in the two models are also same. The corruption of restricted \awtppd protocol is by an additive error, while in 1-SMT-PD  the adversary can  arbitrarily modify the   $|S|=t$ wires. 
However for restricted \awtppd protocol since  $S_r=S_w=S$, modifying $t$ components $(c_{i_1}, \cdots c_{i_t})$ to $(c'_{i_1}, \cdots c'_{i_t})$ is equivalent to ``adding" the error $e$ with $\mathsf{SUPP}(e) =S$ and $(e_{i_1}, \cdots e_{i_t}) = ((c'_{i_1}-c_{i_1}), \cdots (c'_{i_t} -c_{i_t}))$ and so for these channels additive errors cover all possible adversarial }
\end{enumerate}

Next, we show the lower bound of transmission rate  for 
Using Theorem \ref{the_relation1}, for a \smtone over $N$ wires and $t=\rho N$, 
there is a corresponding restricted $(\epsilon, \delta)$-\awtppd protocol whose rate  is upper bounded by,
\[
R\leq 1-\rho+2\epsilon (1+\log_{|\Sigma|}\frac{1}{\epsilon})+2\epsilon n
\]
Since the transmission rate of a 1-$(\epsilon, \delta)$-SMT protocol is the inverse of the
rate of the corresponding restricted $(\epsilon, \delta)$-\awtppd protocol, we have 
\[
\begin{split}
\mathsf{TR} &= \frac{1}{R}\\
&\geq \frac{1}{1-2\rho+2\epsilon (1+\log_{|\cal W|}\frac{1}{\epsilon})+2\epsilon n}\\
&= \frac{N}{N-2t+2N\epsilon (1+\log_{|\cal W|}\frac{1}{\epsilon})+2\epsilon nN}\\
\end{split}
\]

Last, we show the lower bound on the message round of the  \smtone protocol. 
Since \smtone protocol is a special case of $(\epsilon, \delta)$-SMT-PD protocol, and it was shown that the lower bound on message round complexity for $(\epsilon, \delta)$-SMT-PD protocol is at least three, the lower bound of \smtone protocol is also three.

\end{proof}

\subsection{Detail of bounding $c_1$ and $c_2$}\label{ap_sec6}

\begin{proof}

We use the notations  in \cite{GGO10}.

From \cite{GGO10}, we have $\log |\cW_i|\geq N$, $N=\frac{K}{1-D}$, $K= \frac{k_{\min}}{n-t}+\lambda$, $k_{\min}= m$, and $\log |\cM|=m$. This gives the information rate,
\[
\frac{\log |\cM|}{\sum_{i=1}^n \log |\cW_i|}=\frac{m}{n N}= \frac{m}{n\frac{1}{1-D}(\frac{m}{n-t}+\lambda)}.
\] 
Let $\xi >0$ be a small constant. 
Choose $\lambda =\frac{n^2}{\xi}$, $D=\xi$, and $m=\frac{n^2}{\xi^2}(n-t)$. 

So the information rate is,
\[
\frac{\log |\cM|}{\sum_{i=1}^n \log |\cW_i|}=\frac{m}{\frac{n}{1-\xi}\frac{m}{n-t}(1+\xi)}\geq 1-\frac{t}{n}-2\xi 
\] 
 Let $n_0$ be an integer   that satisfies $n_0\geq \frac{1}{\xi}$ and $\frac{1}{e}\leq \sqrt[\leftroot{-3}\uproot{3}n_0^2]{\frac{1}{n_0^2}}$.
The decoding error is for $n\geq n_0$ is,
\[
\delta =t(1-D)^\lambda\leq n(1-\xi)^{\frac{n^2}{\xi}}\overset{(1)}{=}n(\frac{1}{e})^{n^2}=\frac{n}{n^2}\leq \xi,
\]
where (1) is from $(1-\xi)^{\frac{1}{\xi}}\rightarrow \frac{1}{e}$ as $\xi\rightarrow 0$. 

That is the information rate of protocol I \cite{GGO10}  approaches  $1-\frac{t}{n}-\xi$ as the number of wires $n$ approaches infinity. 
 
\vspace{3mm}

Secondly, we show the bound of $c$.

From \cite{GGO10}, we have  $\log |\cW_i|\geq N+K$, $N=2K$, $K \geq \frac{r}{n-t}$, and $\log |\cM|=r$. This implies, 
\[
\frac{\log |\cM|}{\sum_{i=1}^n \log |\cW_i|}=\frac{r}{n N}\leq \frac{r}{3n\frac{r}{n-t}}=\frac{1}{3}(1-\frac{t}{n})
\] 
So there is $c\leq \frac{1}{3}$. 

It implies the information rate of protocol II \cite{GGO10} is approximate to $c(1-\frac{t}{n})$ as the number of wires $n$ is approximate to infinity, with $c\leq \frac{1}{3}$.

\end{proof}

\end{appendix}

\end{document}